\documentclass[11pt]{article}

  \usepackage{palatino}
  \usepackage{mathpazo}
  \usepackage{braket}
  \usepackage{amsfonts}
  \usepackage{amssymb}
  \usepackage{amsmath}
  \usepackage{latexsym}
  \usepackage{amsthm}
  \usepackage{hyperref}
  \usepackage[boxed]{algorithm2e}
\usepackage{authblk}

\hypersetup{pdfpagemode=UseNone}

\newtheorem{theorem}{Theorem}[]
\newtheorem{lemma}[theorem]{Lemma}

\newtheorem{definition}[theorem]{Definition}

\newtheorem{remark}[theorem]{Remark}

    \date{}

                 \newcommand{\falsedef}[1]{\vspace{0.1cm} \noindent
               \textbf{ #1. }}

  \newcommand{\parag}[1]{~\\\noindent \textbf{#1}}
  \newcommand{\comment}[1]{}
  
  \newcommand{\tinyspace}{\mspace{1mu}}
  
  \newcommand{\abs}[1]{\left\lvert\tinyspace #1 \tinyspace\right\rvert}
  
  \newcommand{\norm}[1]{\left\lVert\tinyspace #1 \tinyspace\right\rVert}
  \newcommand{\tr}{\operatorname{Tr}}

  \newcommand{\class}[1]{\textup{#1}}

  \newcommand{\half}{\ensuremath{\frac{1}{2}}}

  \newcommand{\inner}[2]{\langle #1, #2 \rangle}

  \newcommand{\calO}{\mathcal{O}}

  \newcommand{\C}{\mathbb{C}}
  \newcommand{\R}{\mathbb{R}}

  \newcommand{\dsum}{\displaystyle\sum}

  \newcommand{\ketbra}[2]{\ket{#1} \bra{#2}}
  \newcommand{\kb}[1]{\ketbra{#1}{#1}}
  \newcommand{\bracket}[2]{\langle #1 | #2 \rangle}

  \newcommand{\trnorm}[1]{\norm{#1}_{\mathrm {tr}}}

  \newcommand{\pnorm}[1]{\norm{#1}_{\mathrm{p}}}

  \newcommand{\ayes}{A_{\rm yes}} 
  \newcommand{\ano}{A_{\rm no}}

  \def\({\left(}
  \def\){\right)}

  \def\poly{\textup{poly}}

  \usepackage{color,graphicx}

  \newcommand{\NP}{\class{NP}} 
   
  \newcommand{\PP}{\class{PP}} 
  \newcommand{\QMA}{\class{QMA}} 
  \newcommand{\QCMA}{\class{QCMA}} 
  \newcommand{\oSQMA}{\class{oSQMA}} 
  \newcommand{\OSQMA}{\class{oSQMA}} 
  \newcommand{\SQMA}{\class{SQMA}}
  \newcommand{\SQMAone}{\ensuremath{\class{SQMA}_1}}

  \title{\bf $\QMA$ with subset state witnesses}
  \author[1]{Alex B. Grilo}
  \author[1,2]{Iordanis Kerenidis}
  \author[2]{Jamie Sikora}
  \affil[1]{ IRIF,
    CNRS, Universit\'{e} Paris
  Diderot, Paris, France }
  \affil[2]{Centre for Quantum Technologies, National University of Singapore,
  Singapore}

\begin{document}
  \maketitle

  \begin{abstract} 
    The class \QMA{} plays a fundamental role in quantum complexity theory
    and it has found surprising connections to condensed matter physics
    and in particular in the study of the minimum energy of quantum
    systems. In this paper, we further investigate the class \QMA{} and its
    related class \QCMA{} by asking what makes quantum witnesses potentially
    more powerful than classical ones. We provide a definition of a new
    class, \SQMA{}, where we restrict the possible quantum witnesses to the
    "simpler" subset states, i.e. a uniform superposition over the elements
    of a subset of $n$-bit strings. Surprisingly, we prove that this class
    is equal to \QMA{}, hence providing a new characterisation of the class
    \QMA{}. We also prove the analogous result for $\QMA(2)$ and describe a new complete problem for \QMA{} and a stronger
    lower bound for the class $\QMA_1$.

  \end{abstract}

  %------------------------------------------------------------------------------------%
  %------------------------------------------------------------------------------------%
  \section{Introduction}
  \label{sect:Intro}
  %------------------------------------------------------------------------------------%
  %------------------------------------------------------------------------------------%

  One of the notions at the heart of classical complexity theory is the class
  $\NP$ and the fact that deciding whether a boolean formula is satisfiable or not
  is $\NP$-complete \cite{Cook71, Levin73}. The importance of
  $\NP$-completeness became apparent through the plethora of combinatorial problems
  that can be cast as constraint satisfaction problems and shown to be
  $\NP$-complete.
  Moreover, the famous PCP theorem
  \cite{AroraLMSS98, AroraS98} provided a new, surprising description of
  the class $\NP$: any language in $\NP$ can be verified efficiently by
   accessing probabilistically a constant number of bits of a polynomial-size
  witness. This opened the way to showing that in many cases, approximating the
  solution of $\NP$-hard problems remains as hard as solving them
  exactly. An equivalent definition of the PCP theorem states that it remains $\NP$-hard to decide
  whether an instance of a constraint satisfaction problem is satisfiable or any assignment
  violates at least a constant fraction of the constraints. 

  Not surprisingly, the quantum analog of the class $\NP$, defined by Kitaev
  \cite{KitaevSV02} and called $\QMA$, has been the subject of extensive study in
  the last decade. Many important properties of this class are known,
  including a strong amplification property and an upper bound of
  $\PP$ \cite{MarriottW05}, as well as numerous complete problems related to the ground state energy of different types of Hamiltonians
  \cite{KitaevSV02, KR03, OT05, CM13, HallgrenNN13}. Nevertheless, there are still many open questions about the class $\QMA$, including whether it admits perfect completeness or not. 

  Moreover, it is still wide-open if a quantum PCP theorem exists. One way to
  phrase the quantum PCP theorem is that any problem in $\QMA$ can be verified
  efficiently by a quantum verifier accessing a constant number of qubits of a
  polynomial-size quantum witness. Another way would be that the problem of
  approximating the ground state energy of a local Hamiltonian within a constant
  is still $\QMA$-hard. There have been a series of results, mostly negative,
  towards the goal of proving or disproving the quantum PCP theorem, but there is
  still no conclusive evidence \cite{AharonovAV13}.

  Another important open question about the class $\QMA$ is whether the witness
  really need be a quantum state or it is enough for the polynomial-time quantum
  verifier to receive a classical witness. In other words, whether the class
  $\QMA$ is equal to the class $\QCMA$, which is the class of problems that are
  decidable by a polynomial-time quantum verifier who receives a
  polynomial-size classical witness. Needless to say, resolving this question
  can also have implications to the quantum PCP theorem, since in case the two
  classes are the same, the quantum witness can be replaced by a classical one,
  which may be more easily checked locally. In addition, we know that perfect
  completeness is achievable for the class $\QCMA$ \cite{JordanKNN12}. 

  In this paper, we investigate the class $\QMA$ by asking the following simple, yet fundamental question: what makes a quantum witness potentially more powerful than a classical one? Is it the fact that to describe a quantum state one needs to specify an exponential number of possibly different  amplitudes?
   Is it the different relative phases in the quantum state? Or is it something else altogether? 

  \parag{$\QMA$ with subset state witnesses.}
  We provide a definition of a new class, where we restrict the quantum witnesses to be as "classical" as possible, without having by definition an efficient classical description (otherwise our class would be trivially equal to $\QCMA$). All definitions and statements of the results are made formal in their respective sections. 

  For any subset $S \subseteq [d]$, we define the subset state $\ket{S} \in \C^d$, as the uniform superposition over the elements of $S$. More precisely, 
  $ \ket{S} := \frac{1}{\sqrt{|S|}} \sum_{i \in S} \ket{i} \textrm{.} $

  \falsedef{The class \SQMA{}}
    A promise problem  $A = (\ayes, \ano)$ is in $\SQMA{}$ ({\bf
    S}ubset state {\bf \QMA{}}) if for every $x \in \ayes \cup \ano$, there
    exists a polynomial time quantum verifier $V_x$, such that
    \begin{itemize}
      \item \textit{(completeness)} for all $x \in \ayes$, there exists a subset state witness
        $\ket{S}$, such that the verifier accepts with probability at least $2/3$.
      \item \textit{(soundness)} for all $x \in \ano$ and all quantum witnesses $\ket{\psi}$,
        the verifier accepts with probability at most $1/{3}$.
    \end{itemize}
 
  The only difference from $\QMA$ is that in the yes-instances, we ask that there exists a {\em subset state witness} that makes the quantum verifier accept with high probability. In other words, an honest prover need only provide such subset states, which in principle are conceptually simpler. 

  Notice, nevertheless, that the Group Non-Membership Problem is in $\SQMA$, since the witness in the known $\QMA$ protocol is a subset state \cite{Watrous00}. Moreover, we can define a version of our class with two non-entangled provers, similarly to $\QMA(2)$, and we can again see that the protocol of Blier and Tapp \cite{BlierT09} which shows that any language in $\NP$ has a $\QMA(2)$ proof system with logarithmic size quantum messages uses such subset states. Hence, even though the witnesses we consider are quite restricted, some of the most interesting containments still hold for our class.

  Even more surprisingly, our main result shows that $\SQMA$ is, in fact, equal to $\QMA$ and the same for the two-prover case. 

\quad \\ 
\textbf{Result 1.} 
\quad 
$\SQMA = \QMA$ \quad and \quad $\SQMA(2)=\QMA(2)$. \\ 
 
  Hence, for any problem in $\QMA$, the quantum witness can be a subset state. This provides a new way of looking at $\QMA$ and shows that if quantum witnesses are more powerful than classical ones, then this relies solely on the fact that a quantum witness can, in some sense, convey information about an arbitrary subset of classical strings through a uniform superposition of its elements. On the other hand, one way to prove that classical witnesses are as powerful as quantum witnesses, is to find a way to replace such subset states with a classical witness, possibly by enforcing more structure on the accepting subset states. 

  Our proof relies on a geometric lemma, which shows, {for instance}, that for any unit vector in $\C^{2^n}$, there exists a subset state, such that their inner product is  $\Omega(1/{\sqrt n})$. This lemma, in conjunction with standard amplification techniques for $\QMA$ imply our main result. 

  \parag{Complete problems.} 
  The canonical $\QMA$-complete problem is the following: Given a Hamiltonian acting on an $n$-qubit system, which is a sum of "local" Hamiltonians each acting on a constant number of qubits, decide whether the ground state energy is at most $a$ or all states have energy at least $b$, where $b-a \geq 1/poly(n)$. The first question is whether we can show that the same problem is complete if we look at the energy of any subset state instead of the ground state. In fact, we do not know how to show that this problem is complete: when we try to follow Kitaev's proof of completeness and approximate his {\em history state} with a subset state, we cannot retain a sufficient energy gap. Moreover, there exist Hamiltonians with a low energy ground state, but the energy of all subset states is close to 1. 
 
  In this work, we provide one new complete problem for \QMA{} related to
  subset states. This problem is based on the $\QCMA$-complete
  problem Identity Check on Basis States~\cite{WocjanJB03}. 
  
\quad \\ 
\textbf{Result 2.}
\quad 
The following Basis State Check on Subset States problem is $\QMA$-complete: 
  \begin{itemize}
    \item Input: Let $x$ be a classical description of a quantum circuit $Z_x$ on $m$ qubits and
$y$ be an $m'$-bit string, where $n:=|x|$ and $m' \leq m$. Given the
  promise that $x$ satisfies one of the following cases for some 
polynomial\footnote{This polynomial needs to have degree at least that of
$m$ (see Theorem~\ref{thm:bscss} for a formal statement).} $q$, decide which is true: 
  \item Yes: there is a subset $S$ such that $\norm{(\bra{y} \otimes I) Z_x
    \ket{S}}_2^2 \geq 1 - 1/q(n), $
  \item No: for all subsets $S$, we have  
    $ \norm{(\bra{y} \otimes I) Z_x \ket{S}}_2^2 \leq 1/q(n)$.
  \end{itemize}

  \parag{Perfect completeness.}
  Another important open question about $\QMA$ is whether it admits
  perfect completeness. Using our characterisation, this question can be
  reduced to the question of whether $\SQMA$ is equal to $\SQMA_1$. On one hand,
  the result of \cite{Aaronson09} can be used to show that there exists a quantum oracle $A$ relative to which these two classes are not equal, i.e., $\SQMA^A \neq \SQMA_1^A$. On the other hand, 
  proving perfect completeness for $\SQMA$ may be an easier problem to solve, since unlike $\QMA$, the amplitudes involved in the subset states are much easier to handle. 
  Even though we are unable to prove perfect completeness for $\SQMA$, we prove perfect completeness for the following closely related class.

\quad \\ 
\textbf{The class \oSQMA.} 
    A promise problem  $A = (\ayes, \ano)$ is in $\oSQMA{}$ 
    ({\bf o}ptimal {\bf S}ubset state {\bf \QMA{}})
    if for every $x \in
   \ayes \cup \ano$,
    there exists a polynomial time quantum verifier $V_x$, such that
    \begin{itemize}
      \item \textit{(completeness)} for all $x \in \ayes$, there exists a subset state witness
        $\ket{S}$ that maximizes the probability the verifier accepts and this probability is at least ${2}/{3}$.
      \item \textit{(soundness)} for all $x \in \ano$ and all quantum witnesses $\ket{\psi}$,
        the verifier accepts with probability at most $1/{3}$.
    \end{itemize}

  This class still contains the Group Non-Membership problem, while its two-prover version has short proofs for $\NP$.  It remains open to understand whether demanding that a subset state is the optimal witness, instead of just an accepting one, reduces the computational power of the class. Moreover, these two classes coincide in the case of perfect completeness, since all accepting witnesses are also optimal. We prove that the class $\oSQMA$ admits perfect completeness, which implies a stronger lower bound for the class $\QMA_1$ than the previously known $\QCMA$ bound.

\quad \\ 
\textbf{Result 3.} 
\quad  
  $\SQMA_1 = \OSQMA_1=\OSQMA$ \quad and hence, \quad $\OSQMA \subseteq \QMA_1 \subseteq \QMA$. \\ 

  The fact that for the class $\oSQMA$ there exists a subset state which is an optimal witness implies that 
  the maximum acceptance probability is rational and moreover, it is the maximum eigenvalue of the verifier's operator. These two facts
  enable us to extend the rewind technique used by Kobayashi, Le Gall and Nishimura \cite{KobayashiLGN13} and prove our
  result. 
  
  %------------------------------------------------------------------------------------%
  %------------------------------------------------------------------------------------%
  \section{Preliminaries}
  \label{sect:Prelim}
  %------------------------------------------------------------------------------------%
  %------------------------------------------------------------------------------------%

  \subsection{Definitions}
  Let $\Sigma = \{0,1\}$. For $n \in \mathbb{N}$, we define  $[n] := \{1, ..., n\}$.
  The Hilbert-Schmidt or trace inner product between two operators $A$ and $B$ is
    defined as $\inner{A}{B} = \tr(A^\dag B)$.
  For a complex number $x = a + ib$, $a,b \in \R$, we define its norm $|x|$ by $\sqrt{a^2 + b^2}$.
  For a vector $\ket{v} \in \C^d$, its $p$-norm is defined as ${\pnorm{\ket{v}} := \left(\sum_{1
  \leq i \leq d} |v_i|^p\right)^\frac{1}{p}}$.
  For an operator $A$, the trace norm is $\trnorm{A} := \tr{\sqrt{A^\dag A}}$, which is the sum of the singular values of $A$. 

We now state two identities which we use in our analysis.  For normalized
$\ket{v}, \ket{w} \in \C^d$, we have 
\begin{equation} \label{UnitVectorIdentity1}
\max_{0 \preceq C \preceq I}
|\inner{C}{\kb{v} - \kb{w}}| =  \frac{1}{2} \trnorm{\kb{v} - \kb{w}},  
\end{equation} 
since $\kb{v} - \kb{w}$  has largest eigenvalue $\lambda \geq 0$, smallest
eigenvalue $-\lambda$, and the rest are $0$, and the trace norm of a Hermitian
matrix is the sum of the absolute values of its eigenvalues.  We also have that
for $\ket{v}, \ket{w} \in \C^d$,
    \begin{equation} \label{UnitVectorIdentity2}
  \trnorm{\kb{v} - \kb{w}} = 2 \sqrt{1 - |\bracket{v}{w}|^2}. 
  \end{equation}

  \subsection{Complexity classes and complete problems}
  \label{ssec:complexityclasses}
  We start by defining the known quantum complexity classes we will study and a complete problem. 

  \begin{definition}[\QMA]\label{def:QMA}
      A promise problem $A=(\ayes,\ano)$ is in \class{QMA} if and only if there exist polynomials $p$, $q$ and a polynomial-time uniform family of quantum circuits $\set{Q_n}$, where $Q_n$ takes as input a string $x\in\Sigma^*$ with $\abs{x}=n$, a $p(n)$-qubit quantum state, and $q(n)$ ancilla qubits in state $\ket{0}^{\otimes q(n)}$, such that:
      \begin{itemize}
      \item (completeness) If $x\in\ayes$, then there exists a $p(n)$-qubit quantum state $\ket{\psi}$ such that $Q_n$ accepts $(x,\ket{\psi})$ with probability at least $2/3$.
      \item (soundness) If $x\in\ano$, then for any $p(n)$-qubit quantum state $\ket{\psi}$, $Q_n$ accepts $(x,\ket{\psi})$ with probability at most $1/3$.
      \end{itemize} 
  \end{definition}

We can restrict \QMA{} in order to always accept yes-instances, a property called
\emph{perfect completness}.
  \begin{definition}[$\QMA_1$]\label{def:QMA1}
      A promise problem $A=(\ayes,\ano)$ is in $\class{QMA}_1$ if and only if there
      exist polynomials $p$, $q$ and a polynomial-time uniform family of quantum
      circuits $\set{Q_n}$, where $Q_n$ takes as input a string $x\in\Sigma^*$
      with $\abs{x}=n$, a $p(n)$-qubit quantum state, and $q(n)$ ancilla qubits in state $\ket{0}^{\otimes q(n)}$, such that:
      \begin{itemize}
      \item (completeness) If $x\in\ayes$, then there exists a $p(n)$-qubit quantum state $\ket{\psi}$ such that $Q_n$ accepts $(x,\ket{\psi})$ with probability exactly $1$.
      \item (soundness) If $x\in\ano$, then for any $p(n)$-qubit quantum state
        $\ket{\psi}$, $Q_n$
        accepts $(x, \ket{\psi})$ with probability at most $1/3$.
      \end{itemize} 
  \end{definition} 

Another way we can restrict \QMA{} is only allowing classical witnesses, resulting in
the definition of the class \QCMA{} (sometimes also referred to as \class{MQA}~\cite{W09_2,GSU13}).
  \begin{definition}[\QCMA]\label{def:QCMA}
      A promise problem $A=(\ayes,\ano)$ is in \class{QCMA} if and only if there
      exist polynomials $p$, $q$ and a polynomial-time uniform family of quantum
      circuits $\set{Q_n}$, where $Q_n$ takes as input a string $x\in\Sigma^*$
      with $\abs{x}=n$, a $p(n)$-bit string, and $q(n)$ ancilla qubits in state $\ket{0}^{\otimes q(n)}$, such that:
      \begin{itemize}
      \item (completeness) If $x\in\ayes$, then there exists a $p(n)$-bit string
        $y$ such that $Q_n$ accepts $(x,y)$ with probability at least $2/3$.
      \item (soundness) If $x\in\ano$, then for any $p(n)$-bit string $y$, $Q_n$
        accepts $(x, y)$ with probability at most $1/3$.
      \end{itemize} 
  \end{definition} 

  We state here one \QCMA{}-complete problem, the \emph{Identity Check on Basis
  States} problem~\cite{WocjanJB03}. 

\begin{definition}[Identity Check on Basis States~\cite{WocjanJB03}] 
  \label{def:icbs}
Let $x$ be a classical description of a quantum circuit $Z_x$ on $m$ qubits.
Given the promise that $Z_x$ satisfies one of the following cases for $\mu - \delta \geq 1/\poly(|x|)$, decide which
one is true:
\begin{itemize}
\item either there is a binary string $z$ such that 
$|\bra{z} Z_x \ket{z}|^2 \leq 1 - \mu$, 
i.e., $Z_x$ does not act as the identity on the basis states, 
\item or for all binary strings $z$, 
$|\bra{z} Z_x \ket{z}|^2 \geq 1 - \delta$, 
i.e., $Z_x$ acts ``almost'' as the identity on the basis states. 
\end{itemize} 
\end{definition} 

We also consider the two (unentangled) provers version of $\QMA$, defined below.

  \begin{definition}[\QMA(2)]\label{def:QMAtwo}
      A promise problem $A=(\ayes,\ano)$ is in \class{QMA}(2) if and only if there exist polynomials $p$, $q$ and a polynomial-time uniform family of quantum circuits $\set{Q_n}$, where $Q_n$ takes as input a string $x\in\Sigma^*$ with $\abs{x}=n$, two unentangled $p(n)$-qubit quantum states, and $q(n)$ ancilla qubits in state $\ket{0}^{\otimes q(n)}$, such that:
      \begin{itemize}
      \item (completeness) If $x\in\ayes$, then there exists two unentangled $p(n)$-qubit quantum states $\ket{\psi}$ and $\ket{\phi}$ such that $Q_n$ accepts $(x,\ket{\psi},\ket{\phi})$ with probability at least $2/3$.
      \item (soundness) If $x\in\ano$, then for any two unentangled $p(n)$-qubit quantum states $\ket{\psi}$ and $\ket{\phi}$, $Q_n$ accepts $(x,\ket{\psi},\ket{\phi})$ with probability at most $1/3$.
      \end{itemize} 
  \end{definition}

%------------------------------------------------------------------------------------%
%------------------------------------------------------------------------------------%
\section{Subset state approximations}  
%------------------------------------------------------------------------------------%
%------------------------------------------------------------------------------------%

In this section
we state and prove the Subset State Approximation Lemma which
intuitively says that any quantum state can be well-approximated by a subset
state, defined below. 

\begin{definition}
For a subset $S \subseteq [d]$, a \emph{subset state}, denoted here as $\ket{S} \in \C^d$, is a uniform superposition over the the elements of $S$. More specifically, it has the form 
\[ \ket{S} := \frac{1}{\sqrt{|S|}} \sum_{i \in S} \ket{i} \textrm{.} \]
\end{definition}

We now state and prove a useful technical lemma.

\begin{lemma}[Geometric Lemma] 
For a vector $v \in \C^d$, there exists a subset $S \subseteq [d]$ such that 
\[ \frac{1}{\sqrt{|S|}} \left| \sum_{j \in S} v_j \right| \geq \frac{ \norm{v}_2}{8 \sqrt{\log_2(d) + 3}}. \] 
\end{lemma} 

\begin{proof} 
If $v=0$, the lemma statement is trivially true. Suppose $v \in \C^d$ is a nonzero vector and decompose $v$ into real and imaginary parts as $v = u + i w$, where
  $u, w \in \R^d$.  Note that 
\[ 	\norm{v}_2 \leq \norm{u}_2 + \norm{w}_2, 
\]
 by the triangle inequality, implying at least one has norm at
  least $\norm{v}_2/2$. Let us say it is $u$ (the argument
  for $w$ proceeds analogously). 
We now partition $u$ into positive and negative entries such that $u = x - y$ 
where $x, y \geq 0$ and are orthogonal. By the same argument as above, we know at
least one has norm at least $\norm{v}_2/4$. Without loss of generality, {suppose} it is $x$. 
 
Let $T$ denote the support of $x$, i.e., $j \in T$ if and only if $x_j > 0$. 
The idea is to partition $T$ into a small number of sets, where the entries $x_j$ that belong to each set are roughly the same {size}, and the sum of the entries corresponding to one set is a large enough fraction of the norm of the entire vector.

More precisely, let us partition $T$ into the following sets:  
\[ T_k := \left\{j \in T:  \frac{\norm{x}_2}{2^{k}} < x_j \leq \frac{\norm{x}_2}{2^{k-1}} \right\}, \text{ for } k \in [\gamma], \quad T_{\gamma+1} := 
\left\{ j \in T:  0 < x_j \leq \frac{\norm{x}_2}{2^{\gamma}} \right\}   
\] 
for $\gamma := \left\lceil \frac{\log_2(d) + 1}{2} \right\rceil$. We have
\[
\dsum_{j \in \cup_{k \in [\gamma]} T_k} (x_j)^2 = \norm{x}_2^2 - \dsum_{j \in  T_{\gamma+1}} (x_j)^2 \geq \norm{x}_2^2 -d \frac{\norm{x}_2^2}{2^{2\gamma}} =  \norm{x}_2^2 \left(1 - \frac{d}{2^{2\gamma}} \right). 
\]
This implies that there exists $k' \in [\gamma]$ such that 
\[
\dsum_{j \in  T_{k'}} (x_j)^2
\geq 
\frac{ \norm{x}_2^2}{\gamma} \left(1 - \frac{d}{2^{2\gamma}} \right).
\]
Using the definition of $T_k$, we have
\begin{equation*} 
   |T_{k'}| \frac{\norm{x}^2_2}{2^{2(k'-1)}}   \geq \dsum_{j \in  T_{k'}} (x_j)^2
\geq 
\frac{ \norm{x}_2^2}{\gamma} \left(1 - \frac{d}{2^{2\gamma}} \right),
\end{equation*} 
which implies the following lower bound for the size of $T_{k'}$
\begin{equation}\label{GeoProofEqn1}
 |T_{k'}| \geq \frac{  2^{2(k'-1)}  }{\gamma} \left(1 - \frac{d}{2^{2\gamma}} \right).
\end{equation} 
Using again the definition of $T_k$ and Equation~(\ref{GeoProofEqn1}), we have  

\[
\frac{1}{\sqrt{|T_{k'}|}} \sum_{j \in T_{k'}} x_j  \geq \frac{\sqrt{|T_{k'}|} \norm{x}_2}{2^{k'}} \geq \frac{  \norm{x}_2 2^{(k'-1)}  }{2^{k'}\sqrt{\gamma}} \sqrt{ 1 - \frac{d}{2^{2\gamma}}} 
\geq \frac{\norm{v}_2}{8 \sqrt{\log_2(d)+3}}.
\]

Let $S:=T_{k'}$  and $s$ be the vector where $s_j = \frac{1}{\sqrt{|S|}}$ if $j \in S$ and $0$ otherwise. We  have 
\begin{equation*} 
\frac{1}{\sqrt{|S|}} \left| \sum_{j \in S} v_j \right|  
=
|\inner{s}{v}| 
= 
|\inner{s}{u} + i \inner{s}{w}| 
\geq 
|\inner{s}{u}| 
= 
\left| \frac{1}{\sqrt{|S|}} \sum_{j \in S} x_j \right| 
\geq 
\frac{\norm{v}_2}{8 \sqrt{\log_2(d) + 3}}  
\end{equation*} 
as desired. 
\end{proof}

The technique used above of splitting the amplitudes into sets is similar to a
proof in~\cite{JainUW09} which showed a result for approximating bipartite
states by a uniform superposition of  their Schmidt basis vectors. Note that our
result holds for any state and, since we are concerned with a particular fixed
basis,  we need to deal with arbitrary complex amplitudes.

\begin{lemma}[Subset State Approximation Lemma]
  \label{lemma:approximation}
For any $n$-qubit state $\ket{\psi}$, there is a subset $S \subseteq [N]$, where $N := 2^n$, such that $|\bracket{S}{\psi}| \geq \dfrac{1}{8 \sqrt{n+3}}$.  
\end{lemma} 

\begin{remark}
  \label{rem:power2}
  We can further assume the size of the subset is a power of $2$ and lose at
  most a constant factor in the approximation (equal to $\frac{1}{2}$).
\end{remark}

We also show that this approximation factor is optimal by presenting an $n$-qubit
state $\ket{\psi_n}$, for any $n$, where the above bound is tight (up to
constant factors).  In high level, the state has $2^\ell$
basis states with amplitude $\frac{1}{\sqrt{n}\sqrt{2^\ell}}$, for $0 \leq \ell \leq n$, and hence, each of these $n$ subsets of basis states has only a $1/n$ fraction of the total ``weight'' and the amplitudes between different subsets are sufficiently different.
\begin{lemma}
  \label{thm:tightapproximation}
For any $n$, define the following $n$-qubit state 
  \[ \ket{\psi_n} := 
    \sum_{1 \leq i \leq 2^n - 1} \frac{1}{\sqrt{n}\sqrt{2^{\lfloor\log{i}
  \rfloor}}} \ket{i}. \]
Then we have that $\bracket{\psi_n}{S} \leq \frac{2+ \sqrt{2}}{\sqrt{n}} $, for all $S \subseteq [2^n]$. 
\end{lemma}

\begin{proof} 
We see that the amplitudes are non-increasing and thus a subset state that would approximate it the best would be of the form $S = [m]$ for some $m \leq 2^n-1$. 
Thus, we prove now that for all $m$, $S = [m]$ gives an approximation of at
most $\frac{ \sqrt{2} + 2}{\sqrt{n}}$. 

Let $k \in \{ 0, 1, \ldots, n-1 \}$ be such that $2^{k} \leq m \leq 2^{k+1} - 1$. 
We see that 
\[ 
\sum_{i=1}^m \frac{1}{\sqrt{n}\sqrt{2^{\lfloor\log{i}\rfloor}}} 
\leq 
\sum_{i=1}^{2^{k+1}-1} \frac{1}{\sqrt{n}\sqrt{2^{\lfloor\log{i}\rfloor}}} 
=
\sum_{t=0}^{k} \frac{2^t}{\sqrt{n}\sqrt{2^{t}}} = 
\sum_{t=0}^{k} \frac{\sqrt{2^t}}{\sqrt{n}} = \dfrac{
\left( 1+\sqrt{2} \right) \left( \sqrt{2^{k+1}} -1 
\right) }{\sqrt{n}}, 
\] 
where the last equality follows from the formula for a truncated geometric series. We have  
\[ 
\dfrac{1}{\sqrt{m}} \sum_{i-1}^m \frac{1}{\sqrt{n}\sqrt{2^{\lfloor\log{i}\rfloor}}} 
\leq 
\dfrac{
\left( 1+\sqrt{2} \right) \left( \sqrt{2^{k+1}} -1 
\right) }{\sqrt{n} \sqrt{2^k}}
\leq 
\dfrac{
\left( 1+\sqrt{2} \right) \sqrt{2^{k+1}}}{\sqrt{n} \sqrt{2^k}} 
= 
\dfrac{2+\sqrt{2}}{\sqrt{n}},  
\]
as desired. 
\end{proof}  

%------------------------------------------------------------------------------------%
%------------------------------------------------------------------------------------%
\section{Alternative characterisations of \class{QMA} and \class{QMA}(2)}  
%------------------------------------------------------------------------------------%
%------------------------------------------------------------------------------------%

In this section, we prove that $\QMA$ and its two-prover variant can be characterized such that they accept 
subset states. We start by defining formally the new complexity class that is 
by definition contained in $\QMA$.

\begin{definition}[\SQMA]\label{def:SQMA}
  A promise problem $A=(\ayes,\ano)$ is in \SQMA{} if and only if there exist
  polynomials $p$, $q$ and a polynomial-time uniform family of quantum circuits
  $\set{Q_n}$, where $Q_n$ takes as input a string $x\in\Sigma^*$ with
  $\abs{x}=n$, a $p(n)$-qubit quantum state, and $q(n)$ ancilla qubits in state
  $\ket{0}^{\otimes q(n)}$, such that:
    \begin{itemize}
    \item (completeness) If $x\in\ayes$, then there exists a subset $S \subseteq [2^{q(n)}]$ such that $Q_n$ accepts $(x,\ket{S})$ with probability at least $2/3$.
    \item (soundness) If $x\in\ano$, then for any $p(n)$-qubit quantum state $\ket{\psi}$, $Q_n$ accepts $(x,\ket{\psi})$ with probability at most $1/3$.
    \end{itemize} 
\end{definition}
 
\paragraph{Remark.}  
Note that we restricted the witness only in the completeness criterion. In fact, it is straightforward to adapt any $\QMA$ protocol to have a subset state being an optimal witness in the soundness criterion. For example, the prover can send an extra qubit with the original witness and the verifier can measure it in the computational basis. If the outcome is $0$, he continues verifying the proof. If it is $1$, he flips a coin and accepts with probability, say, $1/3$. It is easy to see that an optimal witness for the soundness probability is the string of all $1$'s, which is classical, hence a subset state! Therefore, restricting the proofs in the completeness criterion is the more natural and interesting case. 
 
We prove now that, surprisingly, this restriction does not change the
computational power of $\QMA$.

\begin{theorem} \label{thm:QMAequalsSQMA}
$\QMA = \SQMA$. 
\end{theorem}  
\begin{proof} 
  We have trivially that $\SQMA \subseteq \QMA$ by definition, thus we only need to show that
$\QMA \subseteq \SQMA$. 

Suppose we have a $\QMA$ protocol which verifies a $p(n)$-qubit proof
$\ket{\psi}$ with the two-outcome POVM measurement $\{ C, I - C \}$. More
precisely, without loss of generality, we assume that there exists a polynomial
$r$ such that if $x \in \ayes$, there exists a state $\ket{\psi}$ such that 
$\bra{\psi} C \ket{\psi} \geq 1 - 2^{-r(n)}$ and if $x \in \ano$, we have for
every $\ket{\psi}$, that $\bra{\psi} C \ket{\psi} \leq 2^{-r(n)}$. We show that
the same verification above accepts a subset state with probability at least
$\Omega(1/p(n))$, from which we conclude that the same instance can be decided with
a $\SQMA$ protocol using standard error reduction techniques.
 
If $x \in \ano$ there is nothing to show (since the soundness condition for
$\QMA$ and $\SQMA$ coincide). Suppose $x \in \ayes$ and let $\ket{\psi}$ be a
proof which maximizes the acceptance probability. We then use the Subset State
Approximation Lemma (Lemma~\ref{lemma:approximation}) to approximate
$\ket{\psi}$ with $\ket{S}$, where $S \subseteq [2^{p(n)}]$, satisfies:   
\begin{equation} \label{QMAproofeq1}
|\bracket{\psi}{S}| \geq \dfrac{1}{8 \sqrt{p(n)+3}}
\end{equation} 
We now show that the acceptance probability of $\ket{S}$ is not too small. Note that   
\begin{equation} \label{QMAproofeq2}
\bra{S} C \ket{S}  
= 
\inner{C}{\kb{S}} 
= 
\inner{C}{\kb{\psi}} - \inner{C}{\kb{\psi} - \kb{S}}  
\end{equation} 
and since $\inner{C}{\kb{\psi}} \geq 1 - 2^{-r(n)}$, we concentrate now on bounding $\inner{C}{\kb{\psi} - \kb{S}}$. Clearly, we have 
\begin{equation} \label{QMAproofeq3}
\inner{C}{\kb{\psi} - \kb{S}} 
\leq \max_{0 \preceq C \preceq I}
|\inner{C}{\kb{\psi} - \kb{S}}| =  \frac{1}{2} \trnorm{\kb{\psi} - \kb{S}},  
\end{equation} 
where the last equality comes from Equation~(\ref{UnitVectorIdentity1}).

We now have 
\begin{equation} \label{QMAproofeq4}
\trnorm{\kb{\psi} - \kb{S}} = 2 \sqrt{1 - |\bracket{\psi}{S}|^2} \leq 2 - |\bracket{\psi}{S}|^2, 
\end{equation} 
where the {equality} follows from Equation~(\ref{UnitVectorIdentity2}) and the {inequality} from the fact that, for $x \geq 0$, we have $\sqrt{1-x^2} \leq 1 -
x^2/2$. Combining Equations~(\ref{QMAproofeq1}), (\ref{QMAproofeq2}),
(\ref{QMAproofeq3}), and (\ref{QMAproofeq4}), we have 
\[
\bra{S} C \ket{S} 
\geq 
1 - 2^{-r(n)} - \frac{1}{2} \left( 2 - |\bracket{\psi}{S}|^2 \right) 
=  
\frac{1}{2} |\bracket{\psi}{S}|^2 - 2^{-r(n)}  
\geq 
\frac{1}{128(p(n)+3)} - 2^{-r(n)}. 
\] 
Thus, $\ket{S}$ is accepted with probability $\Omega \left( \dfrac{1}{p(n)} \right)$, as required.  
\end{proof} 

We now define formally the class $\SQMA(2)$. 

  \begin{definition}[\SQMA(2)]\label{def:SQMAtwo}
      A promise problem $A=(\ayes,\ano)$ is in \class{SQMA}(2) if and only if there exist polynomials $p$, $q$ and a polynomial-time uniform family of quantum circuits $\set{Q_n}$, where $Q_n$ takes as input a string $x\in\Sigma^*$ with $\abs{x}=n$, two unentangled $p(n)$-qubit quantum states, and $q(n)$ ancilla qubits in state $\ket{0}^{\otimes q(n)}$, such that:
      \begin{itemize}
      \item (completeness) If $x\in\ayes$, then there exists two  subsets $S, T \subseteq [2^{p(n)}]$ such that $Q_n$ accepts $(x,\ket{S},\ket{T})$ with probability at least $2/3$.
      \item (soundness) If $x\in\ano$, then for any two unentangled $p(n)$-qubit quantum states $\ket{\psi}$ and $\ket{\phi}$, $Q_n$ accepts $(x,\ket{\psi},\ket{\phi})$ with probability at most $1/3$.
      \end{itemize} 
  \end{definition} 
 
\begin{theorem} \label{thm:QMA2equalsSQMA2}
$\QMA(2) = \SQMA(2)$. 
\end{theorem}  
\begin{proof} 
We can use error reduction techniques~\cite{HM10} to assume the
completeness and soundness of the $\QMA(2)$ protocol are $1-2^{-r(n)}$ and
$2^{-r(n)}$, respectively, for some polynomial $r$. If we approximate both
witnesses using subset states and use the same analysis from the one-prover case, we have an inverse
polynomial gap between completeness (using the two subset state witnesses) and the soundness. One can again use the error
reduction techniques from \cite{HM10} since the witnesses (in the reduced error protocol) can be a tensor product of these subset states. 
\end{proof}
 
%------------------------------------------------------------------------------------%
%------------------------------------------------------------------------------------% 
\section{A $\QMA$-complete problem based on subset states}
%------------------------------------------------------------------------------------%
%------------------------------------------------------------------------------------%

In this section, we give a complete problem for $\QMA$ based on circuits mapping subset states to a basis state. This is similar to the $\QCMA$-complete problem
Identity Check on Basis States (see Definition~\ref{def:icbs}).

\begin{definition}[Basis State Check on Subset States (\class{BSCSS}($\alpha$))] 
Let $x$ be a classical description of a quantum circuit $Z_x$ on $m(n)$ input qubits and $a(n)$ ancilla qubits, and
  $y$ be an $m'(n)$-bit string, such that $n:=|x|$, $m$, $a$, and $m'$ are bounded by polynomials and $m' \leq m + a$. Given the promise that $x$ satisfies one of the following cases, decide which is true:
\begin{itemize}
  \item either there exists a subset $S \subseteq [2^{m(n)}]$ such that 
\[ \norm{(\bra{y} \otimes I) Z_x \ket{S} \ket{0}^{\otimes a(n)}}_2^2 \geq 1 - \alpha, \]  
\item or for all subsets $S \subseteq [2^{m(n)}]$, we have  
\[ \norm{(\bra{y} \otimes I) Z_x \ket{S} \ket{0}^{\otimes a(n)}}_2^2 \leq \alpha. \] 

\end{itemize} 
\end{definition} 
  
\begin{theorem} \label{thm:bscss}
For any polynomial $r$, the problem \class{BSCSS} is $\QMA$-complete for $2^{-r(n)} \leq \alpha \leq {\frac{1}{257 (m(n) + 3)}}$.   
\end{theorem} 

Before we prove this result, we first motivate why we study this problem. At
first glance, it looks very similar to a trivial complete problem for $\QMA$.
However, \class{BSCSS} only considers subset states in both the yes and
no-instances, as opposed to arbitrary states for the version for $\QMA$.
Moreover, one may ask what happens to the computational power of $\SQMA$ if one
were to restrict to only rejecting subset states in the soundness criterion in
the definition. The bounds on $\alpha$ in the theorem above give bounds on the
completeness-soundness gap required for this modified definition of $\SQMA$ to
still be equivalent to $\QMA$.
 
To prove the theorem, we show $\SQMA$-hardness and containment in $\SQMA$ separately. The result then follows since $\SQMA = \QMA$. 

\begin{lemma} 
The problem \class{BSCSS} is in $\SQMA$ for $\alpha \leq {\frac{1}{257 (m(n) + 3)}}$. 
\end{lemma} 

\begin{proof} 
The $\SQMA$ verification is as follows. First, the verifier receives a state
$\ket{\psi}$, applies $Z_x$ to $\ket{\psi} \ket{0}^{\otimes a(n)}$, then measures the whole state in the computational
basis to see if the outcome agrees with $y$ on the $m'$ bits. 
 
Suppose we have a yes-instance of \class{BSCSS}. Then we know there exists a
subset state which accepts with probability $1-\alpha$. Now suppose we have a
no-instance of \class{BSCSS}. We know that for all subset states $\ket{S}$, $\norm{(\bra{y} \otimes I) Z_x
\ket{S} \ket{0}^{\otimes a(n)}}_2^2 \leq \alpha$. We now show that
there is no state $\ket{\psi}$ such that $\norm{(\bra{y} \otimes I) Z_x
\ket{\psi}\ket{0}^{\otimes a(n)}}_2^2$ is ``large''. Fix an arbitrary state $\ket{\psi}$ and let
$\ket{S}$ be a subset state with overlap at least $1/(8 \sqrt{m(n)+3})$ from the
Subset State Approximation Lemma (Lemma~\ref{lemma:approximation}). We start
with noticing that
\[ \norm{(\bra{y} \otimes I) Z_x \ket{\psi} \ket{0}^{\otimes a(n)}}_2^2 =
\inner{\tr_{A} \left( (I \otimes \kb{0}_{A}) Z_x^{\dagger} (\kb{y} \otimes I) Z_x (I \otimes \kb{0}_{A}) \right)}{\kb{\psi} }, \]
where $A$ is the $a(n)$-qubit register the ancilla qubits act on.
By a similar analysis as in the proof of Lemma~\ref{thm:QMAequalsSQMA}, we have
that 
\begin{eqnarray*} 
\inner{\tr_{A} \left( (I \otimes \kb{0}_{A}) Z_x^{\dagger} (\kb{y} \otimes I) Z_x (I \otimes \kb{0}_{A}) \right)}{\kb{\psi} } 
& \leq & \alpha + \left(1 -
\frac{1}{128(m(n)+3)}\right) \\
 & = & 1 + \alpha - \frac{1}{128(m(n)+3)}.
\end{eqnarray*} 
Therefore, any proof succeeds with probability at most $1 + \alpha -
\frac{1}{128(m(n)+3)}$. The gap between the 
completeness and soundness is therefore at least 
\[ (1-\alpha)- \left(1 + \alpha -
\frac{1}{128(m(n)+3)} \right) = -2 \alpha + \frac{1}{128(m(n)+3)} = \Omega(1/m(n)), \]
using the assumption that $\alpha \leq \frac{1}{257 (m(n) + 3)}$. We can use standard error reduction techniques to put this protocol into $\SQMA$, as desired. 
\end{proof} 

\begin{lemma} 
For any polynomial $r$, the problem \class{BSCSS} is $\QMA$-hard for $2^{-r(n)} \leq \alpha \leq 1/3$. 
\end{lemma} 

\begin{proof}
Fix a polynomial $r$ and take any $\SQMA$ verification circuit where we assume the following modifications have been made: 
\begin{itemize} 
\item The unitary acts on an $m(n)$-qubit proof and an $a(n)$-qubit ancilla register $A$. 
\item All measurements are deferred until the end of the verification. Denote
  the cumulative unitary the verifier applies as $V$. 
\item We assume the verifier has a special register, $\calO$, at the end
  containing the outcome of the verification. He then measures it in the
  computational basis and accepts on outcome $1$ and rejects on outcome $0$. 
  {
  \item The completeness of the protocol is least $1 - 2^{-r(n)}$ and the soundness is at most $2^{-r(n)}$. 
  }
\end{itemize} 
Then, we define the string $y$ as the single bit $\ket{1}_{\calO}$, i.e., $m' = 1$ here. Then, for the $\SQMA$ protocol, we see the acceptance probability of a state
$\ket{\psi}$ is precisely  
\[ \norm{(\bra{y} \otimes I) V \ket{\psi} \ket{0}^{\otimes a(n)}}_2^2. \]

We now let $(V, y)$ be an instance of \class{BSCSS} with $2^{-r(n)} \leq \alpha \leq 1/3$. We see that the size of the descriptions of $V$ and $y$, as well as $m$, $a$, and $m'$, are at most polynomial in the size of the $\SQMA$ input. It is
clear from the definition of $\SQMA$, that yes-instances of $\SQMA$ are mapped
to yes-instances of the instance of \class{BSCSS} and similarly no-instances
are  mapped to no-instances. Thus, solving this instance of \class{BSCSS}
decides the $\SQMA$ protocol, as desired.  
\end{proof} 

%------------------------------------------------------------------------------------%
%------------------------------------------------------------------------------------%
\section{On the perfectly complete version of $\SQMA$}
\label{sec:pc}
%------------------------------------------------------------------------------------%
%------------------------------------------------------------------------------------%

%  TODO
%  We choose r to be the smallest integer such that 2^r is at least 2q.  The
%  first register now has r qubits.  We apply H to each of the r qubits in the
%  first register, and V to the second register.  We no longer use the third
%  register, nor perform the measurement in step 5.  Instead, we now define the
%  new "acceptance" criterion to be such that either the first register is less
%  than q and the second register is an accepting state of the original
%  verifier, or the first register is at least 2^{r-1} + p.  The acceptance
%  probability becomes exactly 1/2 with this new criterion, because
%   [q/(2^r)] (p/q) + (2^r - 2^{r-1} - p)/(2^r) = p/(2^r) + (2^{r-1} - p)/(2^r).
%   The optimal witness for the original protocol is still the optimal for the
%   new protocol after fixing p and q, and the eigenvector condition is
%   satisfied.  We perform a phase-flip with respect to this new acceptance
%   criterion, and apply H to all the r qubits in the first register, and the
%   inverse of V to the second register.

%p.9 use \mathbb{I} instead of I for readability at several places on the page 

In this section, we study the version of
\SQMA{} with perfect completeness, namely \SQMAone{}.  
Even though we do not prove here that $\SQMA$ admits perfect completeness (i.e., $\SQMA = \SQMA_1$), we
 characterise \SQMAone{} showing that it is equal to a variant of
\SQMA{} where there is an optimal subset state witness.
\begin{definition}[\oSQMA]\label{def:oSQMA}
    A promise problem $A=(\ayes,\ano)$ is in \oSQMA{} if and only if there exist polynomials $p$, $q$ and a polynomial-time uniform family of quantum circuits $\set{Q_n}$, where $Q_n$ takes as input a string $x\in\Sigma^*$ with $\abs{x}=n$, a $p(n)$-qubit quantum state, and $q(n)$ ancilla qubits in state $\ket{0}^{\otimes q(n)}$, such that:
    \begin{itemize}
    \item (completeness) If $x\in\ayes$, then there exists a subset $S \subseteq [2^{p(n)}]$ such that $Q_n$ accepts $(x,\ket{S})$ with probability at least $2/3$ and this subset state \emph{maximizes} the acceptance probability over all states. 
    \item (soundness) If $x\in\ano$, then for any $p(n)$-qubit quantum state $\ket{\psi}$, $Q_n$ accepts $(x,\ket{\psi})$ with probability at most $1/3$.
    \end{itemize} 
\end{definition}

We remark that the perfectly complete versions of $\SQMA$ and $\oSQMA$ coincide since in both cases there is an optimal subset state witness for
yes-instances which leads to acceptance probability $1$. Analogous to the notation for $\QMA_1$ and $\SQMA_1$, we denote the perfectly complete version of $\oSQMA$ as $\oSQMA_1$. 
 
We now state a theorem characterizing $\SQMA_1$ which provides a stronger lower bound for $\QMA_1$.
 
\begin{theorem}
  \label{thm:pc}
  $\SQMA_1 = \OSQMA_1=\OSQMA$ \quad and hence, \quad $\OSQMA \subseteq \QMA_1 \subseteq \QMA .$
\end{theorem}

This theorem is proven using a framework very similar to the works of 
Jordan, Kobayashi, Nagaj and Nishimura~\cite{JordanKNN12} and also
Kobayashi, Le Gall and Nishimura~\cite{KobayashiLGN13}.

The idea is to use the  Rewinding Technique
\cite{Watrous09}\cite{KempeKMV08}\cite{KobayashiLGN13} in order to achieve perfect completeness. For using this
technique we need a quantum circuit that has  maximum acceptance
probability $\half$ for yes-instances.
In order to construct such a circuit we first show that if the $\OSQMA$ verifier uses a specific set of gates, then the maximum acceptance probability for a yes-instance is rational and exactly describable using polynomially many bits.
 
\begin{lemma}
  \label{lemma:sqma_rational}
  If an \oSQMA{} verifier uses only Hadamard,  Toffoli and NOT gates, the maximum acceptance
  probability for
  yes-instances has the form $\frac{p}{q}$, for $p,q \in
  \mathbb{N}$, and $\log{p}, \log{q} \leq l(|x|)$ for some polynomial $l$.
\end{lemma}
\begin{proof}
  As noticed in~\cite{JordanKNN12}, if we apply
  a quantum circuit that consists only of Hadamard, Toffoli, and NOT gates on the computational basis state $\ket{i}$,
  the final superposition will be of the form $\dsum_j
  \frac{k^j_i}{2^\frac{r}{2}} \ket{j}$ for some $k^j_i \in \mathbb{Z}$ and $r \in
  \mathbb{N}$, such that $r$ is polynomially bounded.  Thus, for an optimal \oSQMA{} witness $\frac{1}{\sqrt{|S|}} \dsum_{i
  \in S} \ket{i}$ we have that the final state is $\dfrac{1}{\sqrt{|S|}} \dsum_{i
  \in S} \dsum_j \dfrac{k^j_i}{2^\frac{r}{2}} \ket{j}$.
  Therefore, the maximum acceptance probability is 
  \[ \dfrac{\dsum_{j \in A}
\left(\dsum_{i \in S}k^j_i\right)^2}{|S| 2^r}, 
\]  
\noindent where $A$ is the subset of computational basis states such that the output qubit is
$\ket{1}$ (the measurement outcome where the verifier accepts). It follows that the
  acceptance probability is rational and succinct.
\end{proof} 

In this case, the prover can send a  {classical description} of the maximum
acceptance probability in the original protocol, together with the original
(optimal) subset state proof. Let $a_x$ be the maximum acceptance probability of
the input $x$ in the original protocol and $\frac{p}{q}$ be the value claimed
by the prover to be the maximum
acceptance probability.
The verifier first checks if $\frac{p}{q} \geq \frac{2}{3}$, and rejects
otherwise.

We can create a new Verifier circuit that flips $r+1$ quantum coins, for the
smallest $r$ such that $2^r \geq q$, and either
\begin{enumerate}
  \item accepts with probability $\frac{2^r-p}{2^{r+1}}$,
  \item rejects with probability $\frac{2^r-q+p}{2^{r+1}}$, or
  \item runs the original protocol with probability $\frac{q}{2^{r+1}}$.
\end{enumerate}

For a yes-instance, the prover is honest and sends the correct value
$\frac{p}{q} = a_x$. The maximum success probability of the new Verifier circuit is exactly 
\[ \dfrac{2^r - p}{2^{r+1}} + \dfrac{q}{2^{r+1}} \cdot \dfrac{p}{q} =
\dfrac{1}{2}. \]
For a no-instance the success probability is at most 
\[\dfrac{2^r - p}{2^{r+1}} + \dfrac{q}{2^{r+1}} \cdot a_x \leq \dfrac{2^r - p}{2^{r+1}} + \dfrac{q}{2^{r+1}} \cdot \dfrac{1}{3} \leq
\dfrac{2^r}{2^{r+1}} - \frac{q}{3\cdot 2^{r+1}} \leq \dfrac{1}{2} -
\dfrac{1}{12} = \frac{5}{12} \, ,\] 
\noindent where we used the fact that the maximum acceptance probability $a_x$ in the
original protocol is at most $\frac{1}{3}$,  $\frac{p}{q} \geq \frac{2}{3}$ (which can
be verified with probability $1$) and that $2q \geq 2^r$, by definition.

Then, as we said, we can apply the Rewinding Technique with this new Verifier, and achieve
perfect completeness. For this, we change the initial projector to have all zeroes in the coin
register and the corresponding acceptance projector as described above.
For further and explicit details of why such a
protocol attains perfect completeness, we refer the reader to
reference~\cite{KobayashiLGN13} (in particular, Propositions $17$ and $18$).

%------------------------------------------------------------------------------------%
%------------------------------------------------------------------------------------%
\section{Conclusions}
%------------------------------------------------------------------------------------%
%------------------------------------------------------------------------------------%

  Our results provide a new way of looking at the class $\QMA$ and provide some insight on the power of quantum witnesses. It shows that all quantum witnesses can be replaced by the "simpler" subset states, a fact that may prove helpful both in the case of a quantum PCP and for proving that $\QMA$ admits perfect completeness, towards which we have provided some more partial results. Of course, the main question remains open: Are quantum witnesses more powerful than classical ones and if so, why? What we know now, are some things that do not make the quantum witnesses more powerful, for example arbitrary amplitudes or relative phases. 

  We conclude by stating some open problems. First, can we restrict the quantum
  witnesses even further? In addition, even though we
  proved $\SQMA(2) = \QMA(2)$, we are unable to use this result to show a better
  upper bound than $\class{NEXP}$, the best upper bound currently known.
  Also, can we prove perfect completeness for $\QMA$ through our new
  characterisation? Last, can we obtain other complete problems for $\QMA$, possibly related to finding the energy of subset states of local Hamiltonians?

%------------------------------------------------------------------------------------%
%------------------------------------------------------------------------------------%
\section*{Acknowledgements} 
%------------------------------------------------------------------------------------%
%------------------------------------------------------------------------------------%

The authors acknowledge support from a Government of Canada NSERC Postdoctoral Fellowship, the French National Research Agency (ANR-09-JCJC-0067-01), and the European Union (ERC project QCC 306537). Research at the Centre for Quantum Technologies at the National University of Singapore is partially funded by the Singapore Ministry of Education and the National Research Foundation, also through the Tier 3 Grant ``Random numbers from quantum processes,'' (MOE2012-T3-1-009). 

%------------------------------------------------------------------------------------%
%------------------------------------------------------------------------------------%
\bibliographystyle{alpha}
\bibliography{gks}
%------------------------------------------------------------------------------------%
%------------------------------------------------------------------------------------%

\end{document}